\documentclass[12pt]{article}
\usepackage[latin9]{inputenc}
\usepackage{prettyref}
\usepackage{refstyle}
\usepackage{amsmath}
\usepackage{amsthm}
\usepackage{amssymb}
\usepackage{microtype}
\pdfoutput=1

\makeatletter

\AtBeginDocument{\providecommand\thmref[1]{\ref{thm:#1}}}
\AtBeginDocument{\providecommand\eqref[1]{\ref{eq:#1}}}
\AtBeginDocument{\providecommand\lemref[1]{\ref{lem:#1}}}
\AtBeginDocument{\providecommand\factref[1]{\ref{fact:#1}}}
\AtBeginDocument{\providecommand\claimref[1]{\ref{claim:#1}}}

\RS@ifundefined{subsecref}
  {\newref{subsec}{name = \RSsectxt}}
  {}
\RS@ifundefined{thmref}
  {\def\RSthmtxt{theorem~}\newref{thm}{name = \RSthmtxt}}
  {}
\RS@ifundefined{lemref}
  {\def\RSlemtxt{lemma~}\newref{lem}{name = \RSlemtxt}}
  {}

\theoremstyle{plain}
\newtheorem{thm}{\protect\theoremname}
\theoremstyle{definition}
\newtheorem{defn}[thm]{\protect\definitionname}
\theoremstyle{plain}
\newtheorem{lem}[thm]{\protect\lemmaname}
\theoremstyle{plain}
\newtheorem{fact}[thm]{\protect\factname}
\theoremstyle{remark}
\newtheorem{claim}[thm]{\protect\claimname}
\theoremstyle{remark}
\newtheorem{rem}[thm]{\protect\remarkname}
\theoremstyle{plain}
\newtheorem{cor}[thm]{\protect\corollaryname}
\theoremstyle{plain}
\newtheorem{conjecture}[thm]{\protect\conjecturename}

\usepackage{fullpage}

\newref{claim}{refcmd={Claim \ref{#1}}}
\newref{fact}{refcmd={Fact \ref{#1}}}
\newref{def}{refcmd={Definition \ref{#1}}}
\newref{thm}{refcmd={Theorem \ref{#1}}} 
\newref{cor}{refcmd={Corollary \ref{#1}}} 
\newref{lem}{refcmd={Lemma \ref{#1}}} 

\usepackage{url}
\usepackage[pagebackref,colorlinks=true,linkcolor=blue,citecolor=blue]{hyperref}

\makeatother

\providecommand{\claimname}{Claim}
\providecommand{\conjecturename}{Conjecture}
\providecommand{\corollaryname}{Corollary}
\providecommand{\definitionname}{Definition}
\providecommand{\factname}{Fact}
\providecommand{\lemmaname}{Lemma}
\providecommand{\remarkname}{Remark}
\providecommand{\theoremname}{Theorem}

\begin{document}

\newcommand{\tr}{\operatorname{tr}}

\global\long\def\poly{\mathrm{{poly}}}
\global\long\def\polylog{\mathrm{{polylog}}}%
\global\long\def\zo{\mathrm{\{0,1\}}}%

\global\long\def\mo{\mathrm{{-1,1}}}%

\global\long\def\e{\mathrm{\epsilon}}%

\global\long\def\eps{\mathrm{\epsilon}}%

\global\long\def\E{\mathrm{\mathbb{E}}}%

\global\long\def\F{\mathrm{\mathbb{F}}}%

\global\long\def\P{\mathrm{\mathbb{P}}}%

\title{Boosting uniformity in quasirandom groups: fast and simple}
\author{%
Harm Derksen\thanks{Partially supported by NSF grant DMS 2147769.} \\
Northeastern University\\
ha.derksen@northeastern.edu
\and
Chin Ho Lee\\
North Carolina State University\\
chinho.lee@ncsu.edu	
\and
Emanuele Viola\thanks{Supported by NSF grant CCF-2114116.} \\
Northeastern University\\
viola@ccs.neu.edu
}
\maketitle

\begin{abstract}
We study the communication complexity of multiplying $k\times t$
elements from the group $H=\text{SL}(2,q)$ in the number-on-forehead
model with $k$ parties. We prove a lower bound of $(t\log H)/c^{k}$.
This is an exponential improvement over previous work, and matches
the state-of-the-art in the area.

Relatedly, we show that the convolution of $k^{c}$ independent copies
of a 3-uniform distribution over $H^{m}$ is close to a $k$-uniform
distribution. This is again an exponential improvement over previous
work which needed $c^{k}$ copies.

The proofs are remarkably simple; the results extend to other quasirandom
groups.

We also show that for any group $H$, any distribution over $H^{m}$
whose weight-$k$ Fourier coefficients are small is close to a $k$-uniform
distribution. This generalizes previous work in the abelian setting,
and the proof is simpler.
\end{abstract}

\section{Introduction and our results}

Iterated multiplication of elements in a group is a fundamental problem
that has a long history and wide-ranging applications, and is linked
to long-standing open problems. Already in \cite{DBLP:journals/jacm/LiptonZ77}
it has been pivotal to provide space-efficient algorithms for Dyck
languages. Depending on the underlying group, iterated multiplication
is complete for various complexity classes \cite{KrohnMR66,DBLP:journals/siamcomp/McKenzieC87,Barrington89,Ben-OrC92,ImmermanL95,Miles14}.
For example, Barrington's famous result \cite{Barrington89}
shows that it is complete for $\text{NC}^{1}$ if and only if the underlying
group is non-solvable. This in particular disproved a conjecture about
the complexity of majority \cite{BorodinDFP83}. This type of results
has then been taken further in the study of \emph{catalytic computation}
\cite{BuhrmanCKLS14}, leading to other surprising discoveries \cite{BuhrmanCKLS14,DBLP:conf/stoc/CookM20}.

The focus of this paper is on number-on-forehead communication complexity
\cite{CFL83}. For a survey on the communication complexity of group
products, see \cite{viola-SIGACT19}, and see \cite{KuN97,RaoY2019}
for general background on communication complexity. Concretely, the
input is a matrix of $k\times t$ elements $a_{ij}$ from a group
$H$, and the goal is computing $\prod_{j=1}^{t}a_{1j}\cdots a_{kj}$.
There are $k$ collaborating parties, with Party $i$ knowing all
the input except row $i$.

This problem is also linked to central open problems in communication
complexity. Specifically, \cite{GowersV-cc-int-journal} conjectured
that over certain groups this problem remains hard even for $k$ larger
than $\log n$. Establishing such bounds is arguably the most significant
open problem in the area. A number of candidates have been put forward
over the years, but many have been ruled out via ingenious protocols,
e.g.~in \cite{PRS97,BGKL03,BrodyC08,DBLP:journals/cc/AdaCFN15}.
Interestingly, for the iterated-product candidate proposed in \cite{GowersV-cc-int-journal},
no non-trivial protocol is known.

Iterated group products are also candidate for providing strong separations
between randomized and deterministic number-on-forehead communication.
The current bounds (see \cite{viola-SIGACT19}) give a separation
matching one in \cite{BDPW07}. Stronger bounds could simplify and
strengthen the recent exciting separation \cite{DBLP:journals/corr/abs-2308-12451}.

Returning to the problem, we note that its complexity heavily depends
on the underlying group. If it is abelian, then the problem can be
solved with constant communication, using the public-coin protocol
for equality. Over certain other groups a communication lower bound
of $t/c^{k}$ follows via \cite{Barrington89} from the landmark lower
bound in \cite{BNS92} for generalized inner product; cf.~\cite{MilesV-leak}.
However, this bound does not improve with the size of the group. In
particular it is far from the (trivial) upper bound of $t\log H$,
and it gives nothing when $t$ is constant. Motivated by a cryptographic
application, \cite{MilesV-leak} asked whether a lower bound that
grows with the size of the group, ideally $ct\log H$, can be established
over some group $H$.

Gowers and Viola \cite{GowersV-cc-int,GowersV-cc-int-journal} proved
a bound of $(t\log H)/c^{c^{k}}$ for the group $\text{SL}(2,q)$
of $2\times2$ invertible matrices over $\mathbb{F}_{q}$, which enables
the motivating application from \cite{MilesV-leak}. Subsequent work
\cite{DV-group-mix} simplified the proof and generalized it to any
quasi-random group \cite{Gowers08}, see also \cite{GowersV-cc-int-journal,Shalev16}.
While such bounds do grow with the size of the group, thus answering
the question in \cite{MilesV-leak} and enabling the motivating application
in cryptography, the dependency on the number $k$ of parties is weak:
One can only afford $k$ doubly-logarithmic in the input length.

In this work we give an exponential improvement and obtain bounds
of the form $(t\log H)/c^{k}$, thus matching the state-of-the-art
in number-on-forehead communication \cite{BNS92}. As in \cite{GowersV-cc-int-journal},
we prove stronger results that even bound the advantage such
protocols have when the input is promised to multiply to one of two
fixed elements.
\begin{thm}
\label{thm:cc-bound}Let $H=\mathrm{SL}(2,q)$. Let $P\colon H^{k\times t}\to[2]$
be a number-on-forehead communication protocol with $k$ parties and
communication $b$ bits. For $g\in H$ denote by $p_{g}$ the probability
that $P$ outputs $1$ over a uniform input $(a_{i,j})_{i\le k,j\le t}$
such that $\prod_{j=1}^{t}a_{1j}\cdots a_{kj}=g$. For any $k$ and
any two $g,h\in H$, if $t\ge c^{k}$ then $|p_{g}-p_{h}|\le2^{b}\cdot H^{-t/c^{k}}$.
\end{thm}

The high-level proof technique is the same as in \cite{GowersV-cc-int-journal}.
They reduced the problem to \emph{boosting uniformity }over $m$ copies
of $H$.
\begin{defn}
A distribution $p$ over a set $S$ is $\e$-uniform if $|p(x)-1/S|\le\e/S$.
If $S=H^{t}$ and $k\le t$ we say $p$ is $(\e,k)$-uniform if for
any $k$ coordinates, the induced distribution over those coordinates
is $\e$-uniform. We say $p$ is $k$-uniform if it is $(0,k)$-uniform.
\end{defn}

\cite{GowersV-cc-int-journal} showed that if $s$ is a $2$-uniform
distribution over $H^{m}$ then the convolution (a.k.a. component-wise
product) of some $\ell$ independent copies of $s$ is $H^{-m}$-uniform
over the whole space $H^{m}$. Note that such a result is false for
abelian groups -- the convolution can remain only $2$-uniform. Quantitatively,
they show that $\ell=c^{m}$ copies suffice. In the application to
\thmref{cc-bound} one has $m=2^{k}$, which gives the doubly logarithmic
dependence on $k$.

In this work we give a corresponding exponential improvement on the
number of copies required to boost uniformity: we show that in fact
$\ell=m^{c}$ copies suffice. Our proof is remarkably simple, especially
if we start with $3$-uniform distributions, which we note suffices
for \thmref{cc-bound}. (We discuss below extensions to $2$-uniform
and other groups.) We state this result next.
\begin{thm}
\label{thm:boost-3-to-m}Let $H=\mathrm{SL}(2,q)$. Let $p$ be a $3$-uniform
distribution over $H^{m}$. The convolution of $m^{c}$ independent
copies of $p$ is $H^{-m}$-uniform.
\end{thm}

Our approach allows us to double the uniformity, i.e., go from $k$-uniform
to $2k$-uniform using only a \emph{constant }number of convolutions,
independently of $k$, whereas \cite{GowersV-cc-int-journal} would
use $\ge k$ convolutions. This points to a key difference in the
techniques. In \cite{GowersV-cc-int-journal}, boosting uniformity
is achieved by reduction to \emph{interleaved products, }and appears
tailored to going from $2$-uniform to $3$-uniform. Our approach
is different, and simpler, even taking into account the simple proof
of interleaved mixing from \cite{DV-group-mix}. It can be seen as
a $k$-uniform version of the flattening lemmas discovered in \cite{Gowers08,BabaiNP08,GowersV-cc-int-journal}.
In a nutshell, the $k$-uniformity assumption allows us to remove
all ``low-degree'' Fourier coefficients, leaving only those of degree
$>k$. Then the quasi-randomness assumption, combined with the tensor-product
structure of the Fourier coefficients allows us to ``flatten'' distributions
at a rate proportional to $H^{-ck}$, instead of $H^{-c}$ as in previous
work. We note that while using $k$-uniformity to remove low-degree
coefficients is a common proof technique, we are not aware of previous
work where this is done in the non-abelian setting. This might indicate
that our techniques might find other applications, and in general
we advocate a systematic study of non-abelian analogues of the Fourier
toolkit. Another step in this direction is discussed next.

\paragraph{$(\protect\e,k)$-uniformity vs.~$k$-uniformity.}

Extending the classic work \cite{AlonGM03}, Rubinfeld and Xie~\cite{DBLP:journals/rsa/RubinfeldX13}
showed that every almost $k$-uniform distribution over any Abelian
product group is statistically close to some $k$-uniform distribution.
We generalize their result to any product group. Our approach is significantly
simpler. \cite{DBLP:journals/rsa/RubinfeldX13} decomposes the given
$k$-uniform distribution in a real orthogonal basis instead of the
Fourier basis; we show that in fact the same argument can be carried
out directly over the Fourier basis. A critical observation is that
removing Fourier coefficients of a fixed weight from a real function
keeps the function real.

This generalization, in combination with \thmref{boost-3-to-m}, gives
a refinement of \thmref{boost-3-to-m} where the number of copies
is $k^{c}$ and the final distribution is statistically close to a
$k$-uniform distribution (whereas a direct application of \thmref{boost-3-to-m}
would just give an $(H^{-k^{c}},k)$-uniform distribution).

\paragraph{Extensions.}

\thmref{cc-bound} and \thmref{boost-3-to-m} above can be generalized
to any quasi-random group and to distributions which are $2$-uniform.
This can be done by first using the results in \cite{GowersV-cc-int-journal,DV-group-mix}
to boost $2$-uniformity to $v$-uniformity for a sufficiently large
constant $v$ depending on the quasirandomness of the group (for $\text{SL}(2,q)$,
$v=3$ suffices). This requires a number of convolutions that is exponential
in $v$, but since $v$ is constant it can be afforded. After that,
our results kick in and allow to boost faster.

\section{Preliminaries}

In this section we fix some notation, especially about representation
theory.

For a set $X$, we also write $X$ for its size $|X|$. We write $[i]$
for the set $\{0,1,\ldots,i-1\}$. Every occurrence of ``$c$''
denotes a possibly different universal constant. Replacing ``$c$''
with $O(1)$ everywhere is consistent with a common interpretation
of the latter. For a function $f\colon G\to\mathbb{C}$ we denote by $|f|_{2}^{2}$
the un-normalized quantity $\sum_{x\in G}|f(x)|^{2}$.

Next we present the standard framework of representation theory. The
books by Serre \cite{Serre77}, Diaconis \cite{MR964069}, and Terras
\cite{MR1695775} are good references for representation theory and
non-abelian Fourier analysis. The Barbados notes \cite{Wigderson-Barbados10}
or Section~13 of \cite{MR3584096} or \cite{GowersV-group-mix} provide
briefer introductions. The exposition in these sources is not always
consistent, and often has different aims from ours. So let us give
a quick account of the theory that is most relevant for this work.

\paragraph{Matrices.}

Let $M$ be a square complex matrix. We denote by $\tr(M)$
the trace of $M$, by $\overline{M}$ the conjugate of $M$, by $M^{T}$
the transpose of $M$, and by $M^{*}$ the conjugate transpose $\overline{M^{T}}$
(aka adjoint, Hermitian conjugate, etc.). The matrix $M$ is \emph{unitary}
if the rows and the columns are orthonormal; equivalently $M^{-1}=M^{*}$.

We denote by

\[
|M|_{2}^{2}:=\sum_{i,j}|M_{i,j}|^{2}=\tr(MM^{*}).
\]
This is known as the Frobenius norm, or Schatten 2-norm, or Hilbert-Schmidt
operator, etc.

If $M=AB$ we have
\begin{equation}
|M|_{2}^{2}
= \sum_{i,j} \Bigl|\sum_{k}A_{i,k}B_{k,j} \Bigr|^{2}
\le\sum_{i,j}\Bigl(\sum_{k}|A_{i,k}|^{2}\Bigr)\Bigl(\sum_{k}|B_{k,j}|^{2}\Bigr)
=|A|_{2}^{2}|B|_{2}^{2},\label{eq:MAB-HS-1}
\end{equation}
where the inequality is Cauchy--Schwarz.

\paragraph{Representation theory.}

Let $G$ be a group. A \emph{representation }$\rho$ of $G$ with
dimension $d$ maps elements of $G$ to $d\times d$ unitary, complex
matrices so that $\rho(xy)=\rho(x)\rho(y)$. Thus, $\rho$ is a homomorphism
from $G$ to the group of linear transformations of the vector space
$\mathbb{C}^{d}$. We denote by $d_{\rho}$ the dimension of $\rho$.

If there is a non-trivial subspace $W$ of $\mathbb{C}^{d}$ that
is invariant under $\rho$, that is, $\rho(x)W\subseteq W$ for every
$x\in G$, then $\rho$ is \emph{reducible}; otherwise it is \emph{irreducible.}
Irreducible representations are abbreviated \emph{irreps }and play
a critical role in Fourier analysis. We denote by $\widehat{G}$ a
complete set of inequivalent irreducible representations of $G$.

In every group we have
\begin{equation}
\sum_{\rho\in\widehat{G}}d_{\rho}^{2}=G.\label{eq:sum-d-square}
\end{equation}

We have the following fundamental orthogonality principle.
\begin{lem}[Schur's lemma, see \cite{MR964069}, Page 11 or Lemma 2.3.3 in \cite{Wigderson-Barbados10}]
\label{lem:Schur}
Let $\rho,\psi$ be irreps. Then $\E_{x}\rho(x)_{k,h}\overline{\psi(x)}_{i,j}$ is
$0$ unless $\rho=\psi$ and $k=i$ and $h=j$, in which case it is
$1/d_{\rho}$. In particular, $\E_{x}|\rho(x)_{i,j}|^{2}=1/d_{\rho}$.
\end{lem}

Let $f\colon G\to\mathbb{C}$. The \emph{$\rho$-th Fourier coefficient} of
$f$ is
\[
\widehat{f}(\rho):=\E_{x}f(x)\overline{\rho(x)}.
\]

The Fourier inversion formula is then
\[
f(x)=\sum_{\rho\in\widehat{G}}d_{\rho} \tr\bigl(\widehat{f}(\rho)\rho(x)^{T}\bigr).
\]

We define the \emph{convolution} as follows: 
\[
p*q(x):=\sum_{y}p(y)q(y^{-1}x).
\]

Note that if $p$ and $q$ are distributions then $p*q$ is the distribution
obtained by sampling $x$ from $p$, $y$ from $q$, and then outputting
$xy$.

We note that under this normalization we have
\[
\widehat{p*q}(\alpha)=G \cdot \widehat{p}(\alpha)\widehat{q}(\alpha).
\]

Combining this with \ref{eq:MAB-HS-1} we obtain
\begin{equation}
|\widehat{p*q}(\alpha)|_{2}^{2}\le G^{2} \cdot |\widehat{p}(\alpha)|_{2}^{2}|\widehat{q}(\alpha)|_{2}^{2}.\label{eq:norm-conv-Fourier}
\end{equation}

Parseval's identity is
\[
\E f(x)\overline{g(x)}
=\sum_{\rho}d_{\rho} \tr\bigl(\widehat{f}(\rho)\widehat{g}(\rho)^{*}\bigr).
\]

In case $f=g$ this becomes
\[
\E|f(x)|^{2}=\sum_{\rho}d_{\rho} \tr\bigl(\widehat{f}(\rho)\widehat{f}(\rho)^{*}\bigr)=\sum_{\rho}d_{\rho}|\widehat{f}(\rho)|_{2}^{2}.
\]

\begin{fact}[Theorem 10 in Section 3.2 in \cite{Serre77}, or Theorem 9 in \cite{MR964069}] \label{fact:irrep-tensor}
  Any irrep $\rho$ of $H^{n}$ is the tensor product $\otimes_{i=1}^{n}\rho_{i}$ of $n$ irreps $\rho_{i}$ of $H$. In particular, the dimension of $\rho$ is the
product of the dimensions of the $\rho_{i}$.
\end{fact}

For $\rho=\otimes_{i=1}^{n}\rho_{i}$ we denote by $|\rho|$ the number
of $i$ s.t.~$\rho_{i}$ is not the trivial representation $1$.
\begin{defn}[\cite{Gowers08}]
  A group $H$ is $d$-quasirandom if every non-trivial irrep of $H$ has dimension $\ge d$.
\end{defn}

\section{Boosting uniformity}

In this section we prove \thmref{boost-3-to-m}. The proof follows
by repeated applications of the following theorem.
\begin{thm}
\label{thm:boost-k-2k}Let $H=\mathrm{SL}(2,q)$. Let $p$ be a distribution
over $H^{t}$ that is $(H^{-k},k)$-uniform for $k\ge3$ and $m=\lceil(1+c)k\rceil$.
Then the convolution of $c$ independent copies of $p$ is $H^{-m}$-uniform.
\end{thm}

Note for small $k$ we may have $m=k+1$. But if $k\ge c$ then $m$
is a constant factor larger than $k$.

The choice of the error parameter is not too important because it
can be boosted with convolutions:
\begin{lem}[Lemma 3.3 in \cite{GowersV-cc-int-journal}] \label{lem:good-boost}
Let $p$ and $q$ be $(\e,k)$-uniform distributions over $H^{m}$.
Then $p*q$ is $(\e^{2},k)$-uniform.
\end{lem}

\begin{proof}
It is enough to consider the case $m=k$. We have
\[
\bigl|p*q(x)-1/H^{t}\bigr|
=\Bigl|\sum_{y}\bigl(p(y^{-1})-1/H^{t}\bigr) \bigl(q(yx)-1/H^{t}\bigr)\Bigr|
\le \sum_{y}(\e/G)^{2}=\e^{2}/H^{t}. \qedhere
\]
\end{proof}
In the rest of this section, we prove \thmref{boost-k-2k}. The proof
involves an excursion to 2-norms. The main step is the following new
flattening lemma which can be seen as a $k$-wise variant of the flattening
lemmas discovered in \cite{Gowers08,BabaiNP08,GowersV-cc-int-journal}.
\begin{lem}
\label{lem:basic-convolution}Let $p$ be a distribution over $H^{m}$
where $H$ is $d$-quasirandom. Suppose $p$ is $(H^{-k},k)$-uniform.
Then $|p*p-u|_{2}^{2}\le|p-u|_{2}^{2}\cdot2\cdot H^{m-k}d^{-(k+1)}$.
\end{lem}

We need the following couple of claims to go back-and-forth between
$\e$-uniform and 2-norms.
\begin{claim}
\label{claim:L2-to-Linfty}$|p*p-u|_{\infty}\le|p-u|_{2}^{2}$.
\end{claim}

\begin{proof}
$(p*p-1/G)(x)=\sum_{y}(p(y)-1/G)(p(y^{-1}x)-1/G)\le\sum_{x}(p(x)-1/G)^{2}.$
The last inequality is Cauchy--Schwarz.
\end{proof}
\begin{claim}
\label{claim:rep-bound-good-L_2}Let $p$ be an $\e$-uniform distribution
over a group $G$, and let $\rho$ be a non-trivial representation
of $\rho$ with dimension $d_{\rho}$. Then $|\widehat{p}(\rho)|_{2}^{2}\le d_{\rho}\e^{2}G^{-2}$
\end{claim}

\begin{proof}
The LHS is
\begin{align*}
\sum_{i,j}|\widehat{p}(\rho)_{i,j}|^{2}
 &=\sum_{i,j}|\mathbb{E}_{x} \bigl[p(x)\overline{\rho(x)}]_{i,j} \bigr|^2\\
 & =\sum_{i,j} \bigl|\mathbb{E}_{x}\bigl[(G^{-1}+\e_{x}G^{-1})\overline{\rho(x)}\bigr]_{i,j} \bigr|^{2} && \text{ (for some }\e_{x}\text{ with }|\e_{x}|\le\e\text{)} \\
 & =\sum_{i,j} \bigl|\mathbb{E}_{x}\bigl[\e_{x}G^{-1}\overline{\rho(x)}_{i,j}\bigr] \bigr|^2 && \text{ (by \prettyref{lem:Schur} with $\psi:=1$, using that $\rho$ is non-trivial)}\\
 & \le G^{-2} \sum_{i,j}\mathbb{E}_{x} \bigl[\e_{x}^{2} \cdot |\overline{\rho(x)}_{i,j}|\bigr]^{2}\\
 & \le G^{-2}\sum_{i,j}\e^{2}\mathbb{E}_{x}|\overline{\rho(x)}_{i,j}|^{2}\\
 & =G^{-2}\sum_{i,j}\e^{2}/d_{\rho} && \text{ (by \prettyref{lem:Schur}, see ``in particular'' part)}\\
 & =G^{-2}d_{\rho}\e^{2}. \qedhere
\end{align*}
\end{proof}
\begin{proof}[Proof of \lemref{basic-convolution}]
 Write $G$ for the group $H^{m}$. For any distribution $q$ we
have
\[
|q-u|_{2}^{2}=|q|_{2}^{2}-1/G=G\sum_{\rho}d_{\rho}|\widehat{q}(\rho)|_{2}^{2}-1/G=G\sum_{\rho\ne1}d_{\rho}|\widehat{q}(\rho)|_{2}^{2}.
\]
In our case $q=p*p$, and using \ref{eq:norm-conv-Fourier} and the
above equality we bound the RHS by
\[
\le G^{3}\sum_{\rho\ne1}d_{\rho}|\widehat{p}(\rho)|_{2}^{4}\le G^{2}\cdot|p-u|_{2}^{2}\cdot\max_{\rho\ne1}|\widehat{p}(\rho)|_{2}^{2}.
\]
It remains to bound $G^{2}\max_{\rho\ne1}|\widehat{p}(\rho)|_{2}^{2}$.
We consider two cases:

\smallskip

If $|\rho|>k$, then $d_{\rho}\ge d^{k+1}$ by \factref{irrep-tensor},
so we simply use Parseval to bound 
\[
G^{2}|\widehat{p}(\rho)|_{2}^{2}\le G|p|_{2}^{2}/d_{\rho}\le G|p|_{2}^{2}/d^{k+1}.
\]
We also have $|p|_{2}^{2}\le (\max_{x}p(x)) \cdot \sum_{x}p(x)=\max_{x}p(x)\le2/H^{k}$,
because $p$ is in particular $(1,k)$-uniform. Hence, we get a bound
of $G \cdot 2 \cdot H^{-k}d^{-(k+1)}$, as desired.

\smallskip

If $|\rho|\le k$, then restrict to the non-trivial coordinates of
$\rho$. On those coordinates, $p$ induces a distribution that is $H^{-k}$-uniform.
By \claimref{rep-bound-good-L_2}, we have
\[
G^{2}|\widehat{p}(\rho)|_{2}^{2}\le d_{\rho}H^{-2k}.
\]

Note $d_{\rho}\le H^{k/2}$ by \ref{eq:sum-d-square}. Thus, we obtain
a bound of $H^{-1.5k}\le d^{-(k+1)}$.
\end{proof}
We can now present the proof of \thmref{boost-k-2k}.
\begin{proof}[Proof of \thmref{boost-k-2k}.]
 It is known that $H$ is $\ge cH^{1/3}$-quasirandom, a proof can
be found in \cite{DavidoffSV2003elementary}. Hence, the parameter
$d^{-(k+1)}$ in \lemref{basic-convolution} is $\le cH^{-(k+1)/3}\le H^{-ck}$
for any $k\ge3$. Also, we have $|p-u|_{2}^{2}=|p|_{2}^{2}-1/G$. If
$p$ is $(H^{-k},k)$-uniform then $|p|_{2}^{2}\le\max_{x}p(x)\le2/H^{k}$.
Moreover, the uniformity parameter is maintained when taking convolutions.
So one can apply the lemma a constant number of times to drive the
$L_{2}$ norm to $H^{-m}$, and then convolve one more time to obtain
a distribution that is $H^{-m}$-uniform via \claimref{L2-to-Linfty}.
\end{proof}

\section{Proof of \thmref{cc-bound}}

Let $m:=2^{k}$. As noted in \cite{GowersV-cc-int-journal}, an application
of the box norm (Corollary 3.11 in \cite{ViW-GF2}) shows that the
LHS in the conclusion is $\le cH2^{d}$ times the statistical distance
between the uniform distribution over $H^{m}$ and the convolution
of $t$ independent copies of the following distribution $s$ over
$H^{m}$.
\begin{defn}
\label{def:box-norm-dist}Pick $u_{i}^{0},u_{i}^{1}$ for $i\in[k]$
uniformly from $H$. For $x\in[2]^{k}$ the $x$ coordinate $s(x)$
of $s$ is defined to be $\prod_{i\in[k]}u_{i}^{x_{i}}$.
\end{defn}

\begin{claim}
$s$ is $3$-uniform.
\end{claim}

\begin{proof}
Pick a coordinate $i$ s.t.~$x_{i}\ne y_{i}$. W.l.og.~let $i=0$,
$x_{0}=0$, and $y_{0}=1$. Now $z_{0}$ is equal to either $x_{0}$
or $y_{0}$. Assume w.l.o.g.~that $z_{0}=y_{0}$. Consider a coordinate
$j$ where $z_{j}\ne y_{j}$. Assume again w.l.o.g.~that $j=1$.
We can fix all other $u_{i}$ with $i\ge2$ and prove 3-uniformity
just considering those two coordinates. For concreteness, details
follow.

Up to swapping $y$ and $z$ there are only two cases to consider.
The first is
\begin{align*}
x & =00\\
y & =10\\
z & =11.
\end{align*}

In this case we can fix arbitrarily the $u$ corresponding to $y$,
and then $x$ is uniform because of $u_{0}^{0}$ and $z$ because
of $u_{1}^{1}$.

Alternatively,
\begin{align*}
x & =01\\
y & =10\\
z & =11.
\end{align*}

In this case we can similarly fix arbitrarily the $u$ corresponding
to $z$.
\end{proof}
We note that $s$ is not $4$-uniform, again just considering two
coordinates.

To conclude the proof of \thmref{cc-bound}, note that the convolution
of $m^{c}$ copies of $s$ is $(H^{-m},m)$-uniform by \thmref{boost-3-to-m}.
By \lemref{good-boost} the convolution of $t$ copies is then $(H^{-m\cdot t/m^{c}},m)$-uniform,
and the result follows.

\section{From $(\protect\e,k)$-uniform to $k$-uniform}

In this section we prove the following generalization of \cite{DBLP:journals/rsa/RubinfeldX13}.
\begin{thm}
\label{thm:agm}Let $p$ be a distribution on $G=H^{m}$ s.t.~$|\widehat{p}(\rho)|_{2}\le\eps/G$
for every $\rho:$$|\rho|\in[1,k]$. Then $p$ has distance at most
$3(mH)^{2k}\e$ from a $k$-uniform distribution.
\end{thm}

First we note the following converse to \claimref{rep-bound-good-L_2}.
\begin{claim}
\label{claim:fourier-zero-implies-k-uniform}Let $p$ be a distribution
over $H^{m}$. Suppose $\widehat{p}(\rho)=0$ whenever $|\rho|\in\{1,\ldots,k\}$.
Then $p$ is $k$-uniform.
\end{claim}

\begin{proof}
Consider any $k$ coordinates; assume they are the first $k$ w.l.o.g.
The probability of a string $x\in H^{k}$ is
\[
\sum_{y\in H^{m-k}}p(xy).
\]
By the inversion formula and the assumption this is
\[
\sum_{y}\sum_{\rho}d_{\rho} \tr\bigl(\widehat{p}(\rho)\rho(xy)^{T}\bigr)=\sum_{y}\widehat{p}(1)+\sum_{\rho:|\rho|>k}d_{\rho} \tr\Bigl(\widehat{p}(\rho)\sum_{y}\rho(xy)\Bigr)  .
\]
We have $\widehat{p}(1)=\mathbb{E}_{x}p(x)=1/H^{m}$ and so the first
summand is $1/H^{k}$. We show that the second summand is $0$ by
showing that $\sum_{y}\rho(xy)$ is the zero matrix. To verify this,
write $\rho$ as a tensor product of $\rho_{i}$ using \factref{irrep-tensor}.
Then one entry of $\rho(xy)$ is the product of the entries of the
$\rho_{i}$. There is a non-trivial $\rho_{i}$ corresponding to a
$y$ coordinate. The sum over that coordinate of $y$ yields $0$
by \lemref{Schur}.
\end{proof}
\begin{proof}[Proof of \thmref{agm}.]
Let 
\[
\ell(x):=\sum_{\rho:|\rho|\in[1,k]} d_{\rho} \tr\bigl(\widehat{p}(\rho)\rho(x)^{T}\bigr) 
\]
 be the ``low-degree'' part of $p$, and let 
\[
p'(x):=p(x)-\ell(x)=\sum_{\rho:|\rho|\not\in[1,k]}d_{\rho} \tr{\bigl(\widehat{p}(\rho)\rho(x)^{T}}\bigr).
\]
We first observe that $\ell$ and hence $p'$ is real. This is because
the conjugate $\overline{\rho}$ of an irrep is also an irrep, and
over $H^{t}$ the number of non-trivial coordinates is the same. Hence
the sum over $\rho$ is the same as the sum over $\overline{\rho}$.
Moreover, $\widehat{p}(\overline{\rho}) \cdot \overline{\rho} = \overline{\widehat{p}(\rho) \cdot \rho}$.
Therefore, we can write
\[
  2\ell(x)=\sum_{\rho:|\rho|\in[1,k]} \Bigl( d_{\rho} \tr\bigl( \widehat{p}(\rho)\rho(x)^{T} \bigr) + d_\rho \tr \bigl(\overline{\widehat{p}(\rho)\rho(x)^{T}}\bigr) \Bigr).
\]
The expression inside the brackets is real and so $p'(x)$ is real
as well.

Also, $\sum_{x}p'(x)=G \cdot \widehat{p'}(1)=G \cdot \widehat{p}(1)=1$ by \lemref{Schur}.

However, $p'$ may be $<0$ on some $x$. To remedy that, following
previous work, we will ``mix'' $p'$ with the uniform distribution
so that the mixture becomes a distribution. Note that the mixture
is $k$-uniform by \claimref{fourier-zero-implies-k-uniform} as the
low-degree non-trivial Fourier coefficients of both $p'$ and uniform
are zero.

Concretely, let
\[
q:=(1-\beta)p'+\beta\frac{1}{H^{m}}.
\]
for $\beta$ to be determined. Note that $q$ sums to $1$ as $p'$
and $1/H^{m}$ both do.

To determine $\beta$, first note that by definition
\[
p'(x)\ge-\ell(x).
\]
Crudely, $\ell(x) \le (mH)^k \max_{\rho} d_{\rho} |\tr ( \widehat{p}(\rho)\rho(x)^{T} ) |$.
Now, each absolute value is $\le|\widehat{p}(\rho)|_{2}|\rho(x)|_{2}.$
The first term is $\e/G$ by assumption. For the second, we use the
fact that $\rho(x)$ is unitary, and so $|\rho(x)|_{2}=|I|_{2}=\sqrt{d_{\rho}}$.
So $d_{\rho}|\tr(\widehat{p}(\rho)\rho(x)^{T})| \le d_{\rho}^{3/2}\e/G\le H^{3k/4}\e/G$,
using, in the last inequality, \ref{eq:sum-d-square} over the underlying
group $H^{k}$.

Hence, $|\ell(x)|\le(mH)^{2k}\e/G$ and we can set $\beta:=(tH)^{2k}\e/G$
and obtain that $q\ge0$. As remarked earlier, $q$ sums to 1, and
so $q$ is a distribution. It remains to bound the distance between
$p$ and $q$. We have
\[
|p-q|_{1}=\bigl|p-(1-\beta)p'+\beta/G\bigr|_{1}\le|\ell(x)|_{1}+\beta|p'|_{1}+\beta \cdot |1/G|_{1}.
\]
The last two summands are $\beta$ each. The first one is $\le(mH)^{2k}\e=\beta$
by the bound on $|\ell(x)|$ above. Hence the distance is $\le3\beta=3(mH)^{2k}\e$.
\end{proof}

\bibliographystyle{alpha}
\bibliography{OmniBib}

\end{document}